\documentclass[letterpaper, 10pt, conference]{ieeeconf} 
\IEEEoverridecommandlockouts                              
\newtheorem{theorem}{Theorem}[section]
\newtheorem{lemma}[theorem]{Lemma}
\newtheorem{proposition}[theorem]{Proposition}

\newtheorem{definition}[theorem]{Definition}

\usepackage{bm}
\usepackage[noend]{algorithmic}
\usepackage{algorithm}
\usepackage{color}
\usepackage{epsfig}
\usepackage{amssymb}
\usepackage{amsmath}
\usepackage{amsfonts}
\usepackage{url}
\usepackage[tight,footnotesize]{subfigure}

\newcommand{\ignore}[1]{}

\title{\LARGE \bf
Maximizing System Throughput Using Cooperative Sensing in Multi-Channel Cognitive Radio Networks
}

\author{Shuang Li, Zizhan Zheng, Eylem Ekici and Ness B. Shroff
\thanks{Shuang Li is with Department of Computer Science and Engineering, The Ohio State University, Columbus OH 43210,
        {\tt\small li.908@osu.edu}}%
\thanks{Zizhan Zheng, Eylem Ekici and Ness Shroff are with the Department of Electrical and Computer Engineering, The Ohio State University, Columbus OH 43210,
        {\tt\small \{zhengz,ekici,shroff\}@ece.osu.edu}}%
        \thanks{This work has been funded in part by the Army Research Office MURI award W911NF-08-1-0238 and National Science Foundation awards CNS-1065136, CNS-1012700, and CCF-0914912.
}
}

\begin{document}

\maketitle
\thispagestyle{empty}
\pagestyle{empty}

\begin{abstract}
In Cognitive Radio Networks (CRNs), unlicensed users are allowed to
access the licensed spectrum when it is not currently being used by
primary users (PUs). In this paper, we study the throughput maximization problem for a multi-channel CRN where each SU can only sense a limited number of channels. We show that this problem
is strongly NP-hard, and propose an approximation algorithm with a factor at least $\frac{1}{2}\mu$ where $\mu \in [1,2]$ is a system parameter reflecting the sensing capability of SUs across channels and their sensing budgets. This performance guarantee is achieved by exploiting a nice structural property of the objective function and constructing a particular matching.
Our numerical results demonstrate the advantage of our algorithm compared with both a random and a greedy sensing assignment algorithm.
\end{abstract}

\section{Introduction}
\label{sec:intro}
{In the past decade, cognitive radio networks (CRNs) have emerged as a promising solution for achieving better utilization of the frequency spectrum to satisfy the increasing demand of wireless communication resources. In CRNs, secondary users (SUs) are offered the opportunity of accessing the licensed channel when their activities do not cause disruptions for primary user (PU) transmissions. To this end, the Federal Communications Commission (FCC)~\cite{FCC} has opened the broadcast TV frequency bands for unlicensed users such as WLAN and WiFi. Most recently, congressional negotiators have reached the compromise to allow the auction of TV broadcast spectrum to wireless Internet providers \cite{nytimes}. IEEE has announced the IEEE 802.22 wireless network standard~\cite{ieee} that specifies how to utilize the unused resources between channels in the TV frequency spectrum.

To guarantee a high system throughput in a CRN, the main challenge is for the SUs to accurately detect the channel state of PUs while exploiting transmission opportunities over the white space. Sensing inaccuracies may lead to either a {\it false alarm}, where a channel is detected to be occupied when it is actually idle, or a {\it misdetection}, where a channel is detected to be idle when it is actually occupied. While the former hurts SU throughput, the latter hurts both PU and SU throughput. To improve sensing accuracy, {\bf cooperative spectrum sensing} schemes~\cite{Ganesan05, Peh, Mishra06} have been recently developed, where a joint decision is derived from individual observations made by multiple SUs, which effectively alleviates the impact of incorrect individual decisions on throughput by exploiting the spatial diversity of the SUs.

While cooperative sensing improves sensing accuracy, it also incurs sensing and reporting overhead at the SU side, especially when an SU senses multiple channels in a multi-channel CRN. In particular, requiring each SU to sense all the channels in a CRN may lead to long sensing durations, especially when the number of channels is large, which in turn reduces the average throughput of SUs. It is therefore reasonable to put a limit on the maximum sensing duration that an SU can afford, which translates to a budget on the number of channels that an SU can sense. Due to the hardware constraints, this budget could be different for different SUs. In this paper, we study the throughput optimization problem for a multi-channel CRN subject to this sensing constraint.

Various cooperative sensing protocols have been proposed for maximizing system-wide performance metrics such as sensing accuracy~\cite{Peh} and system throughput~\cite{li,Zhang}. However, these works either focus on a single-channel setting~\cite{Peh,li} or allow each SU to sense all the channels~\cite{kim,Fan_Jiang_2010,Zhang}. In particular, an optimal Bayesian decision rule that maps a vector of local binary decisions made at SUs to a global decision on PU activity has been found for maximizing system throughput in a single channel setting~\cite{li}, which achieves significantly better performance than linear rules such as AND, OR, and majority rules. However, a direct extension of the result in~\cite{li} to the multi-channel setting would require each SU to sense all the channels and incur high sensing duration. On the other hand, most works on multi-channel cooperative sensing put no explicit constraint on sensing duration of SUs. Furthermore, these works either use a simple linear decision rule~\cite{Zhang} or require the transmission of the entire local sensing samples or sensing statistics at each SU. In our work, we choose to use a binary decision rule to avoid the high overhead involved in reporting complete local sensing results. However, instead of using a suboptimal linear rule as in~\cite{Zhang}, we use the optimal decision rule proposed in~\cite{li} for each channel.

In this paper, we study the problem of maximizing the system throughput in a multi-channel CRN, by deciding for each channel, a subset of SUs to sense the channel, subject to the sensing budget constraint at each SU. Our main contributions can be summarized as follows:}

\begin{itemize}
\item {We show that the throughput maximization problem is NP-hard in the strong sense and hence does not have a pseudo-polynomial time algorithm unless P = NP.}
\item We prove that the system throughput function satisfies a structural property, and based on this we propose a matching-based algorithm, which achieves an approximation factor at least $\frac{1}{2}\mu$ where $\mu\in [1,2]$ is a system parameter depending on the sensing capability of SUs across channels and their sensing budgets.
\end{itemize}

This paper is organized as follows. The system model and the problem formulation are introduced in Section~\ref{sec:model}. In Section~\ref{sec:NP}, we prove that the optimization problem is NP-hard in the strong sense. We then prove the structural property of the system throughput function, and propose a matching based algorithm in Section~\ref{sec:sol}. In Section~\ref{sec:simu}, numerical results illustrate the performance of our algorithms. The paper is concluded in Section~\ref{sec:con}.

\section{System Model}
\label{sec:model}

In this section, we present the system model in two parts: communication model and cooperative sensing model. Based on the models, we formulate our overall objective, which  is to decide the channel sensing assignment to maximize the overall system throughput.

\subsection{Communication Model}
\label{subsec:com_mod}

We consider a time-slotted cognitive radio network composed of $M$ orthogonal channels {(each corresponding to a PU)\footnote{Our model can be generalized to the scenario where multiple PUs access the same channel.}} and $N$ SUs. An SU may sense multiple PUs depending on its location. When the channel is idle, SUs that do not interfere with each other can transmit over it. Since scheduling and channel assignment for SU transmission are not the focus of this paper, we employ a simple policy: an SU is randomly selected for transmission over each available channel. Our model can readily be extended to practical models where conflict sets for a given interference model are known. We denote the set of SUs by $S = \{s_1,...,s_N\}$ with $|S|=N$, and the set of channels by $C = \{c_1,...,c_M\}$ with $|C|=M$.

\subsection{Cooperative Sensing Model}
\label{subsec:coop_mod}
We assume that a binary decision is made at an SU for each channel it senses. {Let} $P_f^i(k)$ represent the {\bf probability of false alarm}, i.e., the probability that a SU $s_i$ senses channel $k$ to be occupied when in fact it is idle. {Similarly,} $P_m^i(k)$ represents the {\bf probability of mis-detection}, i.e., the probability that $s_i$ senses channel $k$ to be idle when it is actually occupied. Note that SUs outside the sensing range, if selected for sensing, report random sensing results. For instance, $P_m^i(k)=\frac{1}{2}$ and $P_f^i(k)=\frac{1}{2}$ if SU $i$ is outside the sensing range of PU $k$. We assume that these probabilities can be learned using historical data \cite{Fan_Jiang_2010, Ganesan05, kim}. For instance, given the location information of SUs and hardware parameters such as energy detection threshold and time bandwidth product, etc., $P_m^i(k)$ and $P_f^i(k)$ can be calculated accordingly (see Section~\ref{subsec:gen} for an example).

{\bf Multi-Channel Cooperative Sensing}: SUs may sense the licensed channels cooperatively to reduce sensing errors. To encourage cooperative sensing, we assume {that $\sum_i l_i\ge M$}, which is common in cooperative sensing models \cite{Zhang}, thus the expected number of SUs {that} sense a certain channel is at least $1$. The sensing results of individual SUs are assumed to be independent\ignore{, and for the same SU, the sensing results from different channels are assumed to be independent as well}. As mentioned earlier, due to practical constraints, SUs can sense a limited number of channels. We denote $l_i$ as the maximum number of channels that SU $s_i$ can sense in a time slot, $ 0 \leq l_i \leq M$, for all $i=1,\cdots,N$ and let $l_{max}=\max_{i=1}^{N}{l_i}$. Note that $l_i=0$ means that the SU is in not in the sensing range of any channel, thus it cannot do any sensing and only guess the PU state randomly. In cooperative sensing under the multi-channel setting, multiple SUs choose to sense different channels and predict {channel availability} subject to the budget constraint, and different sensing set assignments lead to different system throughput across channels. We consider a centralized system model, where a central controller is responsible for (1) maintaining system parameters for PUs and SUs (2) in each time slot, deciding for each channel, a subset of SUs to sense the
channel, and (3) making a global decision on channel availability based on the local binary decisions of SUs. Let $S_k$ denote the set of SUs that cooperatively sense channel $k$. The set of all feasible channel sensing assignment policies are denoted by $\mathcal P$, and defined as follows.

\begin{definition}
{\bf Feasible assignment policy} $\mathcal{P}$: A set of sensing sets $\{S_1,\cdots,S_M\}$ is a feasible assignment policy if $\sum\limits_{k=1}^{M}{1_{\{s_i\in S_k\}}}\le l_i$ for all $i$, i.e., all SUs must be assigned to at most $l_i$ channels to sense.
\end{definition}

Let $x_i(k)$ denote the observation of channel $k$ by SU $s_i\in S_k$. {Further,} $x_i(k)=1$ represents that $s_i$ observes channel $k$ to be active, while $x_i(k)=0$ represents that $s_i$ observes channel $k$ to be idle. We {let $\boldsymbol{x}(S_k)$ denote} the vector of observations for channel $k$. Let $\Omega=\{0,1\}$, and let $f_{A}: \Omega ^{|A|}\rightarrow \Omega$ denote a general decision rule that maps the local observations made by a set of SUs, $A \subseteq S$, to global decision on channel activity. As the domain of $f_{A}$ will be clear from the context, we drop the subscript and use $f$ instead. This decision rule applies per channel. Let $B(k)$ denote the activity of channel $k$ such that $B(k)=1$ if channel $k$ is occupied, and $B(k)=0$ otherwise. According to the definitions of false alarm and mis-detection, we define the conditional probability of sensing channel $k$ to be idle when it is indeed idle as follows, where vector $\boldsymbol{y}$ denotes a particular instance of an observation vector:
\vspace{-1ex}
{
\begin{eqnarray}
{}&&P(f(\boldsymbol{x}(S_k))=0|B(k)=0)\nonumber \\
\label{eq:00}
&=&\sum\limits_{{\boldsymbol{y}}: f({\boldsymbol{y}})=0}{P({\boldsymbol{x}}(S_k)=\boldsymbol{y}|B(k)=0)},
\end{eqnarray}
}
where
\vspace{-3ex}
{
\begin{eqnarray*}
{}&&P({\boldsymbol{x}}(S_k)=\boldsymbol{y}|B(k)=0)\\
&=& \prod\limits_{{y}_i=1, s_i \in S_k}{P_f^i(k)} \prod\limits_{{y}_j=0, s_j\in S_k}{(1-P_f^j(k))},
\end{eqnarray*}
}

Similarly, we define the conditional probability of sensing channel $k$ to be occupied when it is indeed occupied:
{
\begin{eqnarray}
{}&&P(f(\boldsymbol{x}(S_k))=1|B(k)=1)\nonumber\\
\label{eq:11}
&=&\sum\limits_{{\boldsymbol{y}}: f({\boldsymbol{y}})=1}{P({\boldsymbol{x}}(S_k)=\boldsymbol{y}|B(k)=1)},
\end{eqnarray}
}
where
{
\begin{eqnarray*}
{}&&P({\boldsymbol{x}}(S_k)=\boldsymbol{y}|B(k)=1)\\
&=&\prod\limits_{{y}_i=1, s_i \in S_k}{(1-P_m^i(k))} \prod\limits_{{y}_j=0, s_j\in S_k}{P_m^j(k)}.
\end{eqnarray*}
}

We assume that in each time slot, a control slot $T_c$ is assigned {for cooperative sensing, during which time a central controller collects $P_m^i(k)$ and $P_f^i(k)$ from SUs, determines the channel sensing assignment, collects sensing results from SUs, and notifies an SU per channel to transmit if that channel is cooperatively sensed to be ``idle.'' Note that each SU $i$ only needs to send updates to the central controller of  $P_m^i(k)$, $P_f^i(k)$ when their values change, e.g, when the location of the SU changes. Furthermore, the central controller only needs to compute a new assignment only when $P_m^i(k)$, $P_f^i(k)$ change. We assume $T_c$ to be a constant in the paper. We further assume that SUs {can transmit at the same bit rate over each channel, and normalize this rate to $1$.} SUs are assumed to be always backlogged and only one of them is scheduled over channel $k$ if sensed available in each time slot. Let ${\pi}_0(k)$ denote the probability that channel $k$ is idle, which is assumed to be acquired accurately over time. The capacity of channel $k$ is denoted by $\gamma(k)$ (after normalization), $k=1,\cdots,M$. We {define} ${\theta}_1(k)=(1-T_c){\pi}_0(k)$ and ${\theta}_2(k)=\gamma (k) (1-{\pi}_0(k))$. Following the logic in \cite{li} and extending to the multi-channel case, we define the expected SU throughput over channel $k$ sensed by $S_k$.
{
\begin{eqnarray}
U_k^1(S_k)&:=&(1-T_c)P(B(k)=0,f({\boldsymbol{x}}(S_k))=0)\nonumber\\
\label{eq:su}
&=&{\theta}_1(k) P(f({\boldsymbol{x}}(S_k))=0|B(k)=0)\\
{}&&\mbox{if }S_k\neq \emptyset\mbox{;}\nonumber\\
U_k^1(S_k)&:=&0\mbox{ if }S_k = \emptyset .\nonumber
\end{eqnarray}
}
\noindent where we assume that if $S_k=\emptyset$, no sensing is conducted for channel $k$ and the channel is never accessed. Likewise, the expected throughput of channel $k$ can be represented by
{
\begin{eqnarray}
\label{eq:pu}
U_k^2(S_k)&:=&{\theta}_2(k) P(f({\boldsymbol{x}}(S_k))=1|B(k)=1)\\
{}&&\mbox{if }S_k\neq \emptyset\mbox{;} \nonumber \\
U_k^2(S_k)&:=&{\theta}_2(k)\mbox{ if }S_k = \emptyset . \nonumber
\end{eqnarray}
}
\begin{definition}
\label{def:Uk}
{\bf System throughput}: For a channel assignment $\{S_1,\cdots,S_M\}$, we define the throughput over channel $k$ to be the sum of SU and PU throughput over channel $k$, denoted as $U_k(S_k)=U_k^1(S_k)+U_k^2(S_k)$. The system throughput is defined as $\sum\limits_{k=1}^{M}{U_k(S_k)}$.
\end{definition}

Note that for a given channel sensing assignment, the achievable system throughput is determined by the decision rule $f$. In this paper, we apply the optimal Bayesian decision rule proposed in~\cite{li} to each channel respectively, to obtain the optimal expected system throughput. Formally, for each channel $k$ and an observation vector $\boldsymbol{y}$ by $S_k$, if ${\theta}_2(k)P({\boldsymbol{x}}(S_k)=\boldsymbol{y}|B(k)=1)\ge {\theta}_1(k) P({\boldsymbol{x}}(S_k)=\boldsymbol{y}|B(k)=0)$, the decision on channel $k$ is ``occupied'', and the contribution to throughput is ${\theta}_2(k)P({\boldsymbol{x}}(S_k)=\boldsymbol{y}|B(k)=1)$ ; otherwise, the decision on channel $k$ is ``idle'' and the contribution is ${\theta}_1(k) P({\boldsymbol{x}}(S_k)=\boldsymbol{y}|B(k)=0)$.}

\subsection{Problem Formulation}
\label{subsec:max_formulation}
We formulate the optimization problem to maximize the system throughput, including PUs and SUs on all channels, as follows:

Problem (A):$ \max\limits_{\{S_1,\cdots,S_M\} \in \mathcal{P}}{\sum\limits_{k=1}^{M}{U_k(S_k)}}$,

\noindent where the Bayesian decision rule is implicit in the definition of $U_k(\cdot)$.

Our goal is to decide the optimal channel sensing assignment to maximize system throughput. We adopt a common assumption that PUs can tolerate interference to a certain extent, which may appear in the form of a constraint as in \cite{Fan_Jiang_2010, Peh} and our earlier paper \cite{li} for the single channel setting. In the future, we plan to extend our solution presented in this paper to Problem (A) with explicit constraints on PU throughput.

We assume that the system is static and the optimization is done in a single time slot. Note that the solution of the static assignment would apply to multiple time slots if $P_m^i(k)$ and $P_f^i(k)$ do not change over time, or if changes occur over a much slower time scale.

\ignore{
\subsection{Cooperative Sensing Protocol}
\label{subsec:prot}

We describe the communication and cooperative sensing protocols as follows:

1) SUs report $P_m^i(k)$ and $P_f^i(k)$ to the central controller.

2) The central controller determines the channel sensing assignment $\{S_1,\cdots,S_M\}$ {based on $P_m^i(k)$, $P_f^i(k)$ and other system parameters}.

3) The central controller notifies each SU of the indices of channels they should sense.

4) SUs receiving channel IDs sense the corresponding channel and report a binary result to the central controller. We assume that both the sensing channel and the reporting channel are reliable.

5) The central controller makes the fusion decision on the channel availability based on the sensing results from SUs and schedules SUs for transmission over channels that are sensed to be idle (one SU per channel).

6) If the channel is cooperatively sensed to be busy, no SU is scheduled over the channel.
}
\section{Hardness of the Problem}
\label{sec:NP}

In this section, we will show that Problem~(A) is strongly NP-hard \cite{Vijay}, by a reduction from Product Partition, which is NP-complete in the strong sense \cite{Ng}. The Production Partition problem is defined as follows: Given $N$ positive integers $a_1$, $a_2$, $\cdots$, $a_N$, is there a subset $X \subseteq \mathcal N:=\{1,2,\cdots,N\}$ such that $\prod\limits_{i\in X}{a_i}=\prod\limits_{i\in \mathcal N\setminus X}{a_i}$?

We reduce Product Partition to the following subproblem of Problem~(A), with $M=2$, $P_f^i(1)=P_f^i(2)=0$ for all $i$, $P_m^i(1)=P_m^i(2):=P_m^i$ for all $i$, and $l_i = 1$ for all $i$, $\gamma(1)=\gamma(2):=\gamma$, ${\pi}_0(1)={\pi}_0(2):={\pi}_0$, $(1-T_c){\pi}_0:={\theta}_1$, $\gamma (1-{\pi}_0):={\theta}_2$, and ${\theta}_1 ={\theta}_2$.

Let $(S_1,S_2)$ denote a solution to this subproblem. Without loss of optimality, we can assume $S_1$ and $S_2$ form a partition of the set of SUs, i.e., $S_1 \cup S_2 = S$ and $S_1 \cap S_2 = \emptyset$. The expected system throughput can then be easily determined using the Bayesian rule as $U_1(S_1) = {\theta}_1+{\theta}_2(1-\prod\limits_{s_i \in S_1}{P_m^{i}})$ and $U_2(S_2) = {\theta}_1+{\theta}_2(1-\prod\limits_{s_i \in S_2}{P_m^{i}})$. Problem~(A) then becomes: $\max\limits_{S_1 \subseteq S}{\Big[ 2{\theta}_1+{\theta}_2(2-(\prod\limits_{s_i \in S_1}{P_m^{i}} + \prod\limits_{s_i \in S \backslash S_1}{P_m^{i}}))\Big] }$, which is further equivalent to $\min\limits_{S_1 \subseteq S}{(\prod\limits_{s_i \in S_1}{P_m^{i}} + \prod\limits_{s_i \in S \backslash S_1}{P_m^{i}})}$ since $2{\theta}_1+2{\theta}_2$ is a constant. We then establish the strong NP-hardness of Problem~(A) by showing
that this new problem is strongly NP-hard.

\begin{proposition}
Problem~(A) is strongly NP-hard.
\label{thm:NP}
\end{proposition}

\begin{proof}
By the above argument, it suffices to prove that the subproblem, $\min\limits_{S_1 \subseteq S}{(\prod\limits_{s_i \in S_1}{P_m^{i}} + \prod\limits_{s_i \in S \backslash S_1}{P_m^{i}})}$, is strongly NP-hard. Given an instance of Product Partition with parameters $a_1,\cdots, a_N$, we reduce it to an instance of this subproblem as follows: let $P_m^i=a_i/10^r$, $i = 1,\cdots,N$, where $r$ is the smallest integer such that $P_m^i \leq 1$ for all $i=1,...,N$. This reduction can clearly be done in polynomial time. Furthermore, if there is a subset $X \subseteq \mathcal N$, such that $\prod\limits_{i\in X}{a_i}=\prod\limits_{i\in \mathcal N \setminus X}{a_i} = \sqrt{\prod\limits_{i\in \mathcal N}{a_i}}$, then the optimal solution to the subproblem is $2 \sqrt{\prod\limits_{s_i\in S}{P_m^i}}$, and vice-versa. Hence if there is polynomial time algorithm to the subproblem, the Product Partition problem can be determined in polynomial time as well, which contradicts the fact that Product Partition is strongly NP-complete.
\end{proof}

Since Problem~(A) is strongly NP-hard, no pseudo-polynomial time algorithms exist unless P = NP \cite{Vijay}. We will propose a matching-based approximation algorithm that has theoretical lower bound 
in Section~\ref{sec:sol}.

\section{Approximate Solutions}
\label{sec:sol}

In this section, we propose an efficient approximation algorithm for Problem~(A). \ignore{We first note that Problem~(A) can be reviewed as a welfare maximization problem studied in the context of combinatorial auctions, where the set of SUs correspond to the set of items for sale, and the set of PUs are the bidders, and the throughput function $U_k(S_k)$ models the valuation for the $k$-th bidder when it obtains a subset of items $S_k$. Although the general welfare maximization problem is hard to approximate~\cite{Dobzinski}, it allows efficient approximations when {the utility function (system throughput function in our scenario)} satisfies some structural properties~\cite{Dobzinski,Feige}.} We first prove an upper and a lower bound on system throughput. We then propose a matching-based approximation algorithm. By exploiting the structural properties of the problem and the bounds on system throughput, we show that the algorithm achieves an approximation ratio of at least $\frac{1}{2}\mu$, where $\mu \in [1,2]$ is a system parameter and will be defined later.

\subsection{Property of the System Throughput}
\label{subsec:sub}
We will show the range of the system throughput $U_k(\cdot)$ in the following lemma.

\begin{lemma}
\label{lem:single}
For any SU $s_i$ and channel $c_k$, we have $\theta_1(k)+\theta_2(k)\ge U_k({s_i})\ge \max{\Big\{{\theta}_1(k),{\theta}_2(k)\Big\}}$.
\end{lemma}

\begin{proof}
\vspace{-1ex}
{\small\begin{eqnarray}
U_k(\{s_i\})&=& \max{ \Big\{ {\theta}_1(k)(1-P_f^i(k)), {\theta}_2(k) P_m^i(k) \Big\} }\nonumber\\
& + & \max{ \Big\{ {\theta}_1(k)P_f^i(k), {\theta}_2(k)(1-P_m^i(k)) \Big\} }\nonumber\\
\label{eq:any}
&\ge & \max{\Big\{{\theta}_1(k),{\theta}_2(k)\Big\}} \nonumber
\end{eqnarray}
}

\vspace{-2ex}
Furthermore, it is clear from the definition of $U_k(\cdot)$ that $\theta_1(k)+\theta_2(k)\ge U_k({s_i})$, since at most both PU and SU can achieve their full capacity.
\end{proof}

\ignore{\begin{proposition}
\label{lem:s_subadd}
For all $n>0$ and all set $C$ with $|C|=n$, $\sum_{s_i\in C}{(U_k(\{s_i\})-\Delta_k)}\ge U_k(\cup_{s_i\in C}{\{s_i\}})-\Delta_k$ for all $k$.
\end{proposition}

\begin{proof}
When $n=1$, it trivially holds. We only need to investigate cases where $n\ge 2$.
\begin{eqnarray}
{}&&\sum_{s_i\in C}{(U_k(\{s_i\})-\Delta_k)}- \Bigg[U_k(\cup_{s_i\in C}{\{s_i\}})-\Delta_k\Bigg]\nonumber\\
&=& \sum_{s_i\in C}{U_k(\{s_i\})}-U_k(\cup_{s_i\in C}{\{s_i\}})-(n-1)|\theta_1(k)-\theta_2(k)| \nonumber\\
&\overset{(a)}\ge & n \max{\{\theta_1(k),\theta_2(k)\}}-(\theta_1(k)+\theta_2(k))\nonumber\\
&&-(n-1)|\theta_1(k)-\theta_2(k)| \nonumber\\
&=& (n-2)\min{\{\theta_1(k),\theta_2(k)\}}\ge 0\nonumber
\end{eqnarray}
\noindent where (a) is by Lemma~\ref{lem:single} and the fact that the no utility function can achieve the full SU capacity and the full PU capacity at the same time.
\end{proof}
}

\subsection{A Matching-Based Approximation Algorithm}
\label{subsec:approx}
In this section, we propose a maximum weighted matching (MWM)~\cite{west} based algorithm to Problem~(A). We first provided a detailed description of our algorithm (see Algorithm~\ref{alg:mwm}), and then establish its approximation factor.

\begin{algorithm}[t]
    \caption{{A maximum weighted matching based algorithm for maximizing the system throughput across channels}}\label{alg:mwm}
    {{\small Input: $N$, $M$, $T_c$, ${\pi}_0(k)$, $\gamma(k)$ for all $k$; $l_i$ for all $i$; $P_m^i(k)$, $P_f^i(k)$ for all $i$ and $k$ \\Output: $U$ and $S_k$ for all $k$}}
    \begin{algorithmic}[1]
    {{\small \STATE $S_k\leftarrow \emptyset$ for all $k$
	\STATE $V\leftarrow \{s_1^1,\cdots,s_1^{l_1},\cdots,s_N^1,\cdots,s_N^{l_N}\}\cup \{c_1,\cdots,c_M\}$
	\STATE $E\leftarrow \bigcup_{i=1,\cdots, N;\ k=1,\cdots, M}{\big\{\bigcup_{j=1}^{l_i}{(s_i^j,c_k)}\big\}}$
    \STATE $G\leftarrow (V,E)$
	\STATE $w(s_i^j,c_k)\leftarrow U_k(\{s_i\})$, $\forall i=1,\cdots,N$, $j=1,\cdots,l_i$, $k=1,\cdots, M$
	\STATE $\mathcal M \leftarrow \mbox{a maximum weight matching in }G$
	
		\STATE $S_k \leftarrow \{s_i: (s^j_i,c_k) \in \mathcal M\}$, $\forall k$
		\STATE $R \leftarrow \{s_i^j: s_i^j$ is not matched in $\mathcal M\}$
		\FORALL{$s_i^j\in R$}
			\STATE $k^*\leftarrow \arg\max\limits_{k\in \{1,\cdots,M\}, s_i \not \in S_k}{\Big[{U}_k(S_k\cup \{s_i\})-{U}_k(S_k)\Big]}$
			\STATE $S_{k^*}\leftarrow S_{k^*}\cup \{s_i\}$
		\ENDFOR
		\STATE $U \leftarrow \sum\limits_{k=1}^{M}{U_k(S_k)}$
    \STATE $U_1 \leftarrow \max\limits_{k=1}^{M}{U_k(S)}$
	\IF{$U_1 > U$}
	    \STATE $U \leftarrow U_1$
        \STATE $k^* \leftarrow \arg\max \limits_{k=1}^{M}{U_k(S)}$
		\STATE $S_{k^*}\leftarrow S, S_k \leftarrow \emptyset$ $\forall k \neq k^*$
	\ENDIF}
	}

    \end{algorithmic}
\end{algorithm}

The algorithm starts with constructing a complete and weighted bipartite graph (lines 2-4), where for each channel $k$, a vertex $c_k$ is constructed, and for each SU $s_i$, $l_i$ vertices are constructed corresponding to the $l_i$ copies of the SU, denoted as $s^j_i, j = 1,...,l_i$, and for any pair of vertices $s^j_i$ and $c_k$, there is an edge connecting them. The weight of an edge $(s^j_i, c_k)$ is then defined as $w(s^j_i, c_k) = U_k(\{s_i\})$ (line 5).

{A} maximum weight matching in the bipartite graph is then found (line 6), and for each edge $(s^j_i,c_k)$ in the matching, SU $s_i$ is assigned to sense channel $c_k$. A greedy heuristic is applied for determining the assignment of the remaining copies of SUs to channels (lines 8-12). Basically, the remaining copies are first sorted in an arbitrary order, and a copy of $s_i$ is assigned to the channel that provides the maximum marginal improvement of the system throughput among all the channels not assigned to $s_i$ yet. This scheme is then compared with {another scheme for which} all SUs are assigned to a single channel that gives maximum throughput (line 13). The algorithm outputs whichever {scheme provides a larger system throughput}.

We then analyze the complexity of Algorithm~\ref{alg:mwm}, which is dominated by computing the maximum weighted matching and evaluating the throughput function $U_k(\cdot)$. It is shown in~\cite{li} that for a given sensing set $S_k$, $U_k(S_k)$ can be evaluated using a dynamic programming algorithm in pseudo-polynomial time. Let $Q$ denote the time complexity for one evaluation of $U_k(\cdot)$. Note that the total number of such evaluations is bounded by $Nl_{max}M$. Therefore, the time complexity of Algorithm~\ref{alg:mwm} is $O(Nl_{max}MQ+{(Nl_{max}+M)}^3)$.

To establish the approximation ratio of Algorithm~\ref{alg:mwm}, we first construct a maximal matching called M$\_$Gdy that approximates the MWM and captures two key aspects: 1) SUs may have different sensing abilities for each channel; 2) channels are competing for SUs with limited sensing budget. We first prove a lower bound on the system throughput using the sensing assignments determined by M$\_$Gdy, which is then used to prove the approximation ratio of Algorithm~\ref{alg:mwm}. The matching is constructed as follows: 1) Partition the channel set $C$ into groups indexed by SU, and each group is labeled as $C_i$ that includes all channels $k$ with $U_k(\{s_i\})\ge U_k(\{s_j\})$ where $j\neq i$. Ties are randomly broken. Let $r_i$ denote the size of $C_i$. 2) Sort the channels $k$ in each group $C_i$ by $U_k^0$ in descending order, where $U_k^0=\min_{i \in S}{U_k(\{s_i\})}$. 3) Pick the first $l_i$ channels from each group $C_i$ (the set is labeled as $C_i^{l_i}$) and assign SU $i$ to sense these channels. 4) Randomly assign an unused SU copy to each of the rest channels. We will next show a lower bound on the system throughput using M$\_$Gdy in Lemma~\ref{lem:matching}. We define $\lambda_i=\min\{l_i,r_i\}/r_i$, $\rho_i=\min_{k\in C_i^{l_i}}{\frac{U_k^*}{U_k^0}}$ where $U_k^*=\max_i{U_k(\{s_i\})}$. Note that by Lemma~\ref{lem:single}, $\rho_i \in [1,2]$ for all $i$. We have the following performance bound, where $\mu=1+\min_{i \in S}{\lambda_i (\rho_i -1)}$, and $|M\_Gdy|$ denotes the system throughput using the sensing assignments determined by M$\_$Gdy.

\begin{lemma}
M$\_$Gdy achieves a system throughput no less than $\mu \sum_k{U_k^0}$.
\label{lem:matching}
\end{lemma}
\vspace{1ex}
\begin{proof}
{\small
\begin{eqnarray}
\frac{|M\_Gdy|}{\sum_k{U_k^0}}&\ge&\frac{\sum_{i\in {S}}{\Big[\sum_{k\in C_i^{l_i}}{U_k^*}+\sum_{k\in C_i\setminus C_i^{l_i}}{U_k^0}\Big]}}{\sum_k{U_k^0}}\nonumber\\
&\ge & \frac{\sum_{i\in {S}}{\Big[\sum_{k\in C_i^{l_i}}{U_k^0 \rho_i}+\sum_{k\in C_i\setminus C_i^{l_i}}{U_k^0}\Big]}}{\sum_{i\in {S}}{\sum_{k\in C_i}{U_k^0}}}\nonumber\\
&=& 1+\frac{\sum_{i \in {S}}{\Big[(\rho_i-1)\sum_{k\in C_i^{l_i}}{U_k^0}\Big]}}{\sum_{i\in {S}}{\sum_{k\in C_i}{U_k^0}}}\nonumber\\
&\ge & 1+\frac{\sum_{i \in {S}}{\Big[(\rho_i-1)\lambda_i\sum_{k\in C_i}{U_k^0}\Big]}}{\sum_{i\in {S}}{\sum_{k\in C_i}{U_k^0}}}\nonumber\\
&\ge & 1+\min_{i\in S}{\lambda_i (\rho_i-1)}=\mu\nonumber
\end{eqnarray}
}
\end{proof}

Based on Lemmas~\ref{lem:single} and \ref{lem:matching}, we show the approximation ratio of Algorithm~\ref{alg:mwm} in Proposition~\ref{prop:ratio}.

\begin{proposition}\label{prop:ratio}
Algorithm~\ref{alg:mwm} achieves at least a fraction of $ \frac{1}{2}\mu$ of the optimal system throughput for Problem~(A).
\end{proposition}

\begin{proof}
Let $OPT$ be the optimal solution, and $ALG$ be the solution by Algorithm~\ref{alg:mwm} to Problem~(A). By Lemma~\ref{lem:single}, we know that

{\small
\begin{equation}
\frac{\sum_k{U_k^0}}{OPT}\ge \frac{\sum_{k}{\max\{\theta_1(k),\theta_2(k)\}}}{\sum_k{\theta_1(k)+\theta_2(k)}}\ge\frac{1}{2}.
\label{eq:worst}
\end{equation}
}

\vspace{-2ex}
Since $ALG$ is an outcome at least as good as maximum weight matching and M$\_$Gdy is a matching we construct in a greedy way, we have $ALG\ge |M\_Gdy|$. By Lemma~\ref{lem:matching}, we can achieve $\frac{ALG}{OPT}\ge \frac{1}{2}\mu$.
\end{proof}

\emph{Remark 1}: We note that when $\theta_1(k)\gg \theta_2(k)$ or $\theta_2(k)\gg \theta_1(k)$, we can achieve a solution close to the optimal by Algorithm~\ref{alg:mwm} since Equation~(\ref{eq:worst}) becomes close to $1$. Only when $\theta_1(k)$ and $\theta_2(k)$ for all $k$ are close, Equation~(\ref{eq:worst}) is only right above $\frac{1}{2}$. Also, if SU's sensing abilities across channels vary in a large range, or the sensing budgets of SUs are large, the ratio will be close to $1$ since $\rho_i$, $\lambda_i$ will be large, respectively.

\emph{Remark 2}: In the proof of Proposition~\ref{prop:ratio}, we have ignored the greedy heuristic applied to the copies of SUs not included in the matching. Hence the result established above only provides a lower bound on the performance of our algorithm. Proving a tighter bound for the algorithm that incorporates the greedy heuristic is part of our future work.

\section{Simulations}
\label{sec:simu}
In this section, we study the performance of our algorithm through simulations by comparing Algorithm~\ref{alg:mwm} (MWM) with a random sensing assignment algorithm, and a greedy algorithm (defined next).
In the random algorithm, the copies of SUs are randomly assigned to PUs. The greedy algorithm works as follows: for each PU $k$, the set of SUs are first sorted by $P_m^i(k)+P_f^i(k)$ in a non-decreasing order as its preference list. In each round, a random permutation of the set of PUs is applied. The algorithm then goes through the PU list, and for each PU $k$, a copy of the SU, say $s_i$, with the lowest $P_m^i(k)+P_f^i(k)$ among the remaining SUs, which has not been assigned to $k$ before and has remaining copies, is assigned to $k$. Repeat this procedure till all copies of SUs have been assigned.

\subsection{Simulation Setting}
\label{subsec:gen}
The following parameters are fixed throughout the simulations. We consider a $100 \times 100$ area, where the locations of $M$ PUs are randomly generated. For each PU $k$, its maximum power level is randomly chosen between $1$ and $10$, and ${\pi}_0(k)$ are randomly generated in $[0,1]$. We also set $T_c=0.2$ fixed. In each of the 100 runs of the simulation, we apply the model proposed in~\cite{Sun} to generate $P_m^i(k)$ and $P_f^i(k)$. The details are in our online technical report \cite{tech}.

\begin{figure*}[!t]
\centering
\subfigure[{$M=20$, $l_{max}=3$, $\gamma(k) \sim U[1,3]$.}]{
\includegraphics[scale=0.3]{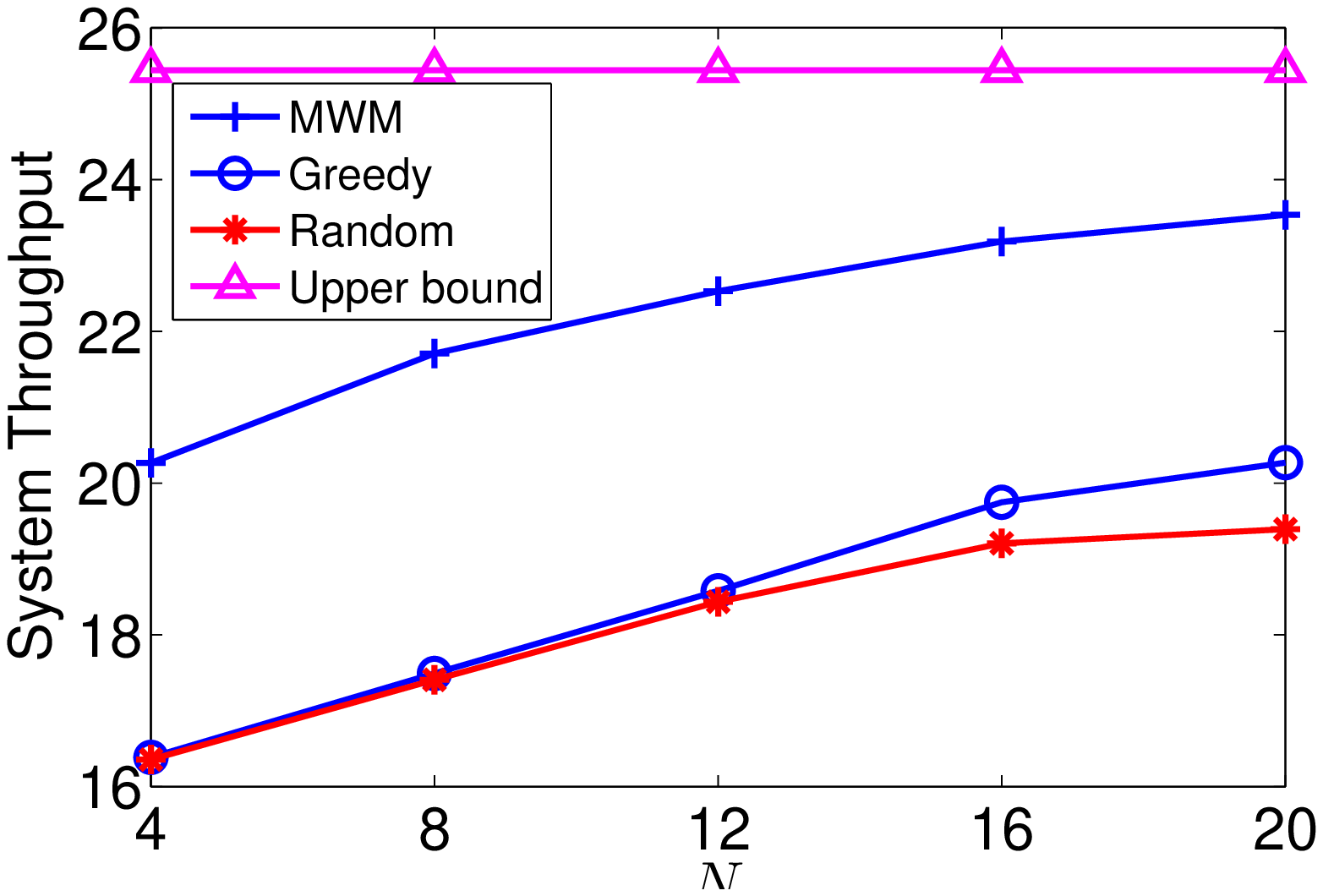}
\label{fig:N}
}
\subfigure[{$M=20$, $N=8$, $\gamma(k) \sim U[1,3]$.}]{
\includegraphics[scale=0.3]{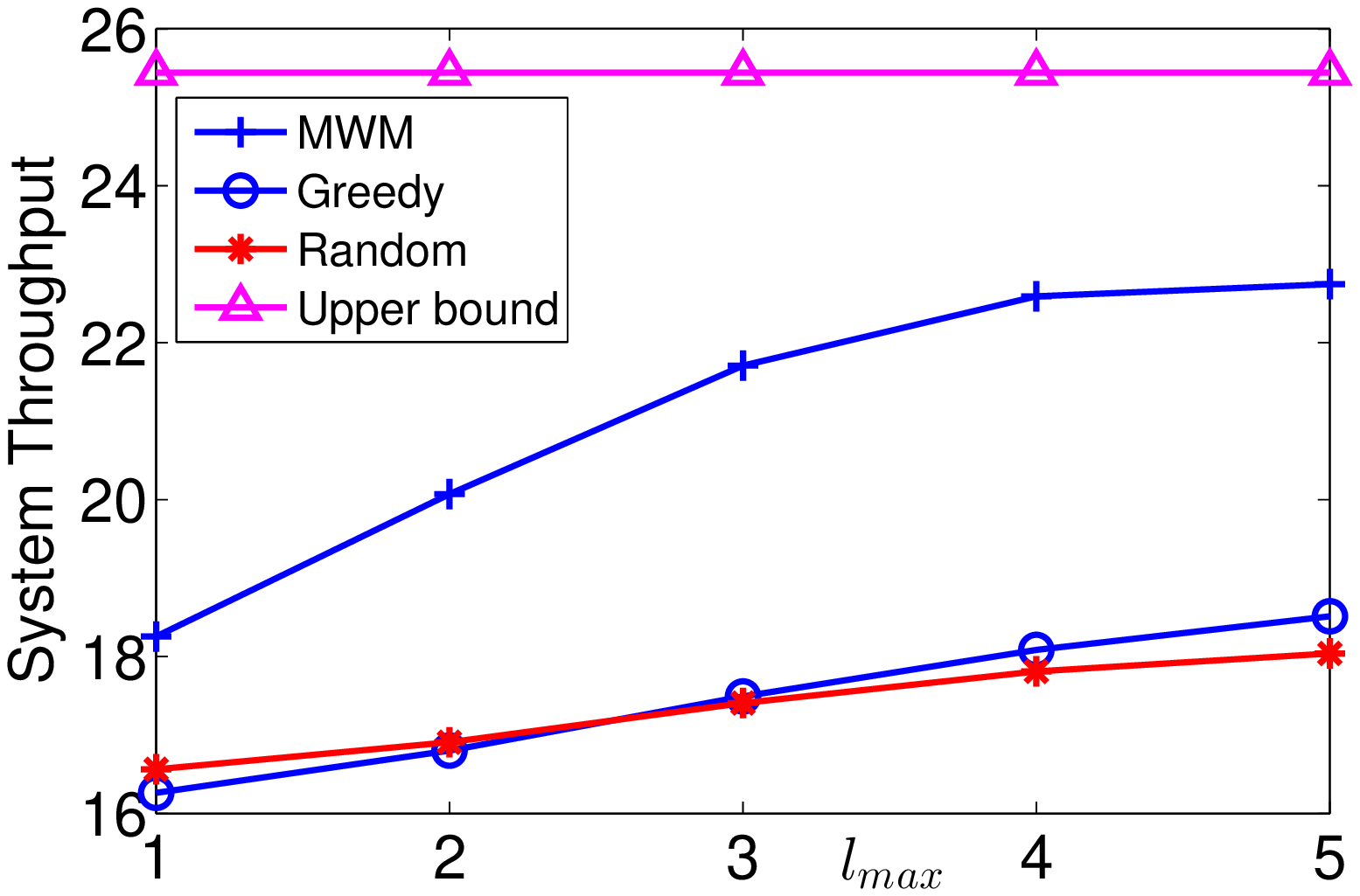}
\label{fig:N_M10}
}
\subfigure[{$M=20$, $N=8$, $l_{max} = 3$.}]{
\includegraphics[scale=0.3]{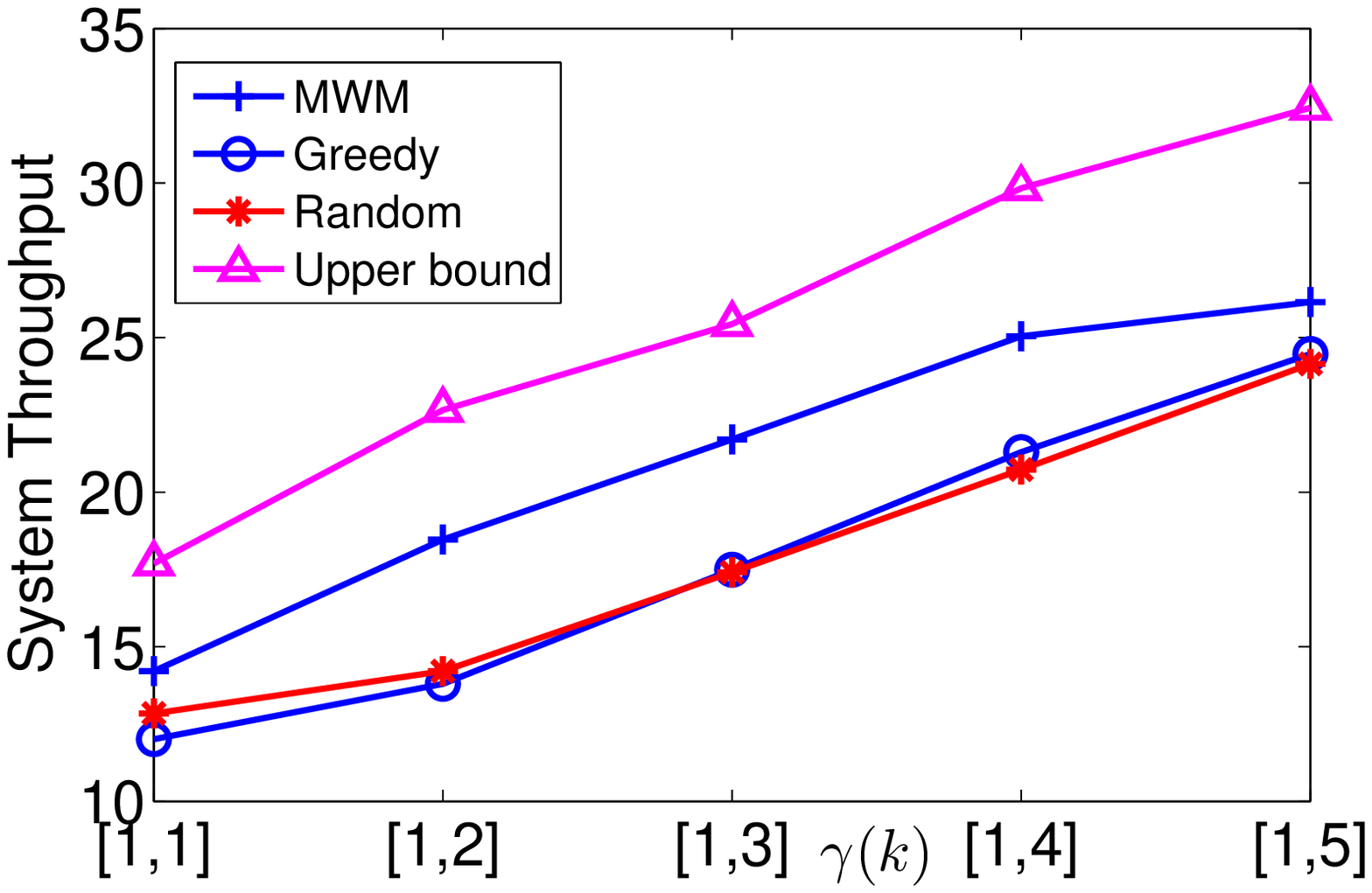}
\label{fig:gamma}
}
\caption{System throughput achieved by our algorithm, greedy algorithm and random algorithm.}
\vspace{-1.5em}
\label{fig:ratio}
\end{figure*}

\subsection{Simulation Results}
\label{subsec:comp}


The simulation results are shown in Figure~\ref{fig:ratio}.
Note that we do not restrict $\sum_i l_i\ge M$ in our simulations. If PU $k$ is not assigned any SU for sensing ($S_k=\emptyset$), the system throughput on channel $k$ is ${\theta}_2(k)$ (Definition~\ref{def:Uk}). In all the figures, we plot $\sum_k{\Big[{\theta}_1(k)+{\theta}_2(k)\Big]}$ as the upper bound for the optimal solution.

In Figure~\ref{fig:N}, we fix $M=20$, $l_{max}=3$, and vary $N$ from $4$ to $20$. For each PU $k$, $\gamma(k)$ in generated randomly in $[1,3]$ and then fixed over all 100 runs. We choose this range since the average PU throughput is usually larger than the unit SU throughput. For each SU $i$, $l_i$ is randomly generated between $1$ and $l_{max}$ and fixed over all the runs. The simulations results are averaged over all 100 runs. We observe that Algorithm~\ref{alg:mwm} achieves significant improvement over the random and the greedy algorithms for all $N$, although the gap shrinks as $N$ increases. For instance, the system throughput of Algorithm~\ref{alg:mwm} is $24\%$ larger than that of the greedy algorithm when $N=4$ and it decreases to $16\%$ when $N=20$. When more SUs join the network, the random and the greedy algorithms have more chance to choose ``good" SUs. The greedy algorithm is comparable to the random algorithm when $N$ is small. However, it wins over the latter when $N\ge 12$. This indicates that the sorting step in the greedy algorithm helps PUs pick the ``right" SUs, which is more useful when $N$ is large. Besides, the performance of Algorithm~\ref{alg:mwm} reaches $95\%$ of the upper bound of the optimal solution when $N=20$.

{In Figure~\ref{fig:N_M10}, $M$, $N$ are fixed to be $20$ and $8$, respectively, and we vary $l_{max}$ from $1$ to $5$. $\gamma(k)$ is again generated randomly in $[1,3]$ and fixed over all 100 runs. Similar to Figure~\ref{fig:N}, Algorithm~\ref{alg:mwm} outperforms both the random and the greedy algorithms, and the greedy algorithm outperforms the random algorithm when $l_{max}\ge 3$. When $l_{max}=4$, the system throughput of Algorithm~\ref{alg:mwm} is $25\%$ better than that of the greedy algorithm, which is the largest gap in the figure. An interesting observation is that the expected number of SU copies when $l_{max}=4$ ($N=8$) is equal to $M=20$, thus every PU is assigned an SU on average. When there is more supply (SUs) than demand (PUs) or more demand than supply, the performance gap between Algorithm~\ref{alg:mwm} and greedy algorithm, random algorithm is not so significant.



In Figure~\ref{fig:gamma}, we fix $M=20$, $N=8$, $l_{max}=3$, and vary the range of the channel capacity $\gamma(k)$. For instance, $[1,2]$ means all channel capacities are randomly generated between $1$ and $2$. Algorithm~\ref{alg:mwm} is constantly better than the other two algorithms. The gap first increases as the channel capacity increases (from $18\%$ to $34\%$) till $\gamma(k) \in [1,2]$, and decreases thereafter ($7\%$ at $\gamma(k) \in [1,5]$). When the channel capacity is comparable to unit SU capacity, the choice of SUs for sensing does not affect the system throughput significantly; When the channel capacity dominates the system throughput, the choice of SUs again loses its leading role. Thus the largest gap appears in the middle.

\section{Conclusion}
\label{sec:con}
In this paper, we investigate the {problem of }throughput maximization using cooperative sensing in multi-channel CRNs, where each SU can only sense a limited number of channels with various sensing capabilities, due to time or energy constraints. We show that under the optimal Bayesian decision rule, the channel sensing assignment problem is strongly NP-hard. A matching based algorithm is then proposed with an approximation ratio that is at least $\frac{1}{2}\mu$ where $\mu\in [1,2]$ is a system parameter.
Our numerical results demonstrate that our algorithm performs significantly better than the
a random channel sensing assignment algorithm and a greedy algorithm. {As part of our future work, we plan to establish a tighter} performance bound for our algorithm enhanced with a greedy heuristic, 
and consider the system throughput maximization problem with extra constraints on the PU throughput.

\bibliographystyle{abbrv}
\bibliography{Multi_Coop_Sens_to_chair}

\begin{thebibliography}{10}

\bibitem{Ng}
{C. T. Ng, M.S.Barketau, T.C.E. Cheng and M. Y. Kovalyov}.
\newblock {``Product Partition'' and related problems of scheduling and systems
  reliability: Computational complexity and approximation}.
\newblock {\em European Journal of Operational Research}, 207(2):601--604, dec
  2010.

\bibitem{Dobzinski}
S.~Dobzinski, N.~Nisan, and M.~Schapira.
\newblock {Approximation algorithms for combinatorial auctions with
  complement-free bidders}.
\newblock In {\em {Proc. of STOC}}, 2005.

\bibitem{Fan_Jiang_2010}
R.~Fan and H.~Jiang.
\newblock Optimal multi-channel cooperative sensing in cognitive radio
  networks.
\newblock {\em IEEE Transactions on Wireless Communications}, 9(3):1128--1138,
  2010.

\bibitem{FCC}
{Federal Communications Commission}.
\newblock Notice of proposed rulemaking, in the matter of unlicensed operation
  in the tv broadcast bands (docket no. 04-186) and additional spectrum for
  unlicensed devices below 900 mhzand in the 3 ghz band (02-380), fcc 04-113.
\newblock May 2004.

\bibitem{Feige}
U.~Feige.
\newblock On maximizing welfare when utility functions are subadditive.
\newblock In {\em Proc. of STOC}, 2006.

\bibitem{Ganesan05}
{G. Ganesan and Y. G. Li}.
\newblock Cooperative spectrum sensing in cognitive radio networks.
\newblock In {\em Proc. of DySPAN}, pages 137--143, nov 2005.

\bibitem{kim}
S.-J. Kim and G.~B. Giannakis.
\newblock Sequential and cooperative sensing for multi-channel cognitive
  radios.
\newblock {\em IEEE Transactions on Signal Processing}, 58(8):4239--4253, aug
  2010.

\bibitem{li}
S.~Li, Z.~Zheng, E.~Ekici, and N.~Shroff.
\newblock Maximizing system throughput by cooperative sensing in cognitive
  radio networks.
\newblock In {\em Proceedings of INFOCOM '12}. IEEE, 2012.

\bibitem{Peh}
E.~Peh and Y.-C. Liang.
\newblock Optimization for cooperative sensing in cognitive radio networks.
\newblock In {\em Proc. of WCNC}, mar 2007.

\bibitem{Mishra06}
{S. M. Mishra, A. Sahai and R. W. Brodersen}.
\newblock Cooperative sensing among cognitive radios.
\newblock In {\em Proc. of ICC}, pages 1658--1663, 2006.

\bibitem{Sun}
C.~Sun, W.~Zhang, and K.~Ben.
\newblock Cluster-based cooperative spectrum sensing in cognitive radio
  systems.
\newblock In {\em Proc. of ICC}, pages 2511 --2515, june 2007.

\bibitem{ieee}
\url{http://grouper.ieee.org/groups/802/22/}.
\newblock {IEEE 802.22, Working Group on Wireless Regional Area Networks
  (WRAN)}.

\bibitem{nytimes}
\url{http://www.nytimes.com/2012/02/17/business/media/congress-to-sell-public-airwaves-to-pay\\-benefits.html?_r=1}.
\newblock {Congress to Sell Public Airwaves to Pay Benefits, New York Times}.

\bibitem{Vijay}
V.~V. Vazirani.
\newblock {\em {Approximation Algorithms}}.
\newblock Springer, mar 2004.

\bibitem{west}
D.~West.
\newblock {\em {Introduction to Graph Theory (2nd Edition)}}.
\newblock {Prentice Hall}, aug 2000.

\bibitem{Zhang}
W.~Zhang and C.~K. Yeo.
\newblock Optimal non-identical sensing setting for multi-channel cooperative
  sensing.
\newblock In {\em Proc. of ICC}, pages 1--5, jun 2011.

\bibitem{tech}
S. Li, Z. Zheng, E. Ekici, and N. Shroff.
\newblock Technical report.
\newblock Maximizing System Throughput Using Cooperative Sensing in Multi-Channel Cognitive Radio Networks, 2012.
\newblock http://www.cse.ohio-state.edu/$\sim$lish/cdc12.pdf

\end{thebibliography}

\end{document}